\documentclass[aps,pra,superscriptaddress,twocolumn,nofootinbib]{revtex4-2}

\usepackage{amsmath,amsthm,amsfonts,amssymb,amscd}
\usepackage{indentfirst}
\usepackage{graphicx}
\usepackage{color}
\usepackage[american]{babel}
\usepackage{float}
\usepackage[dvipsnames]{xcolor}
\usepackage[colorlinks=true,citecolor=ForestGreen,linkcolor=magenta,urlcolor=Mulberry,hyperindex]{hyperref}
\usepackage{mathtools}
\usepackage{bm}
\usepackage{bbold}
\usepackage{units}
\usepackage[nameinlink,capitalize]{cleveref}

\usepackage{physics}

\newcommand{\id}{\mathbb{1}}

\newcommand{\Lcal}{\mathcal{L}}

\newcommand{\Ical}{\mathcal{I}}

\newcommand{\Ecal}{\mathcal{E}}
\newcommand{\Tcal}{\mathcal{T}}

\newcommand{\Wcal}{\mathcal{W}}
\newcommand{\Mcal}{\mathcal{M}}
\newcommand{\Xcal}{\mathcal{X}}
\newcommand{\Zcal}{\mathcal{Z}}

\newcommand{\hil}[1]{\mathcal{H}_\text{#1}}
\renewcommand{\tr}[1]{\text{\normalfont Tr}\left(#1\right)}
\newcommand{\kebra}[1]{|#1\rangle \langle #1 |}

\newcommand{\beq}{\begin{equation}}
\newcommand{\eeq}{\end{equation}}

\newtheorem{theorem}{Theorem}
\newtheorem*{theorem*}{Theorem}

\newcommand{\map}[1]{\widetilde{#1}} 

\begin{document}

\title{Post-measurement states are (very) useful for measurement discrimination} 

\author{Charbel Eid}
\email{charbelezzateid@gmail.com}
\affiliation{Sorbonne Université, CNRS, LIP6, F-75005 Paris, France}

\author{Marco Túlio Quintino}
\email{Marco.Quintino@lip6.fr}
\affiliation{Sorbonne Université, CNRS, LIP6, F-75005 Paris, France}

\begin{abstract}
The standard approach to quantum measurement discrimination is to perform the given unknown measurement on a  probe state, possibly entangled with an auxiliary system, and make a decision based on the measurement outcome obtained. In this work, we go beyond the standard aforementioned scenarios by considering not only the classical outcome of a measurement, but also its post-measurement quantum state. More specifically, instead of considering only the positive-operator valued measure (POVM) operators, we consider their associated Lüders instrument. We prove that, when the post-measurement quantum states are available, the task of discriminating two qubit projective measurements is equivalent to discriminating two copies of quantum states associated to each projector pair, extending previous results known for the case where probe states are separable. Then, we proceed by showing that the advantage of considering post-measurement states in measurement discrimination can be large. We formalize this claim by presenting a family of pairs of measurements where the ratio between the discrimination bias of the measurement discrimination task with and without post-measurement states can be arbitrarily large. This shows that, while the post-measurement state was neglected in most of the previous literature, its use can significantly improve the performance of quantum measurement discrimination.
\end{abstract}

\maketitle

\section{Introduction}
In classical physics and information theory, erroneous identification of a system or object is necessarily the consequence of either incomplete access to the system we are observing or an incomplete knowledge of what we are identifying it as. This is not the case in quantum theory, where one may be provided with complete access to a system along with a guarantee that it is one of two possible systems and complete descriptions of both, yet still be unable to correctly identify which of the two they are observing. The tasks of discriminating quantum objects thus form a fundamental problem, with impact and applications in various branches of quantum information and computation~\cite{2015JPhA...48h3001B, 2009AdOP....1..238B,ChannelD_Chiribellaxx_2008,Optimal-environment-localization-pereira-pirandola,Hypothesis-testing-information-theory-blahut,Unambiguous-discrimination-oracle-operators-chefles-twamley,Quantum-estimation-for-quantum-technology-paris,Perfect-discrimination-of-no-signalling-channels-via-quantum-superposition-of-causal-structures-Chiribella,Characterising-memory-in-quantum-channel-discrimination-via-constrained-separabilit-problems-ohst-quintino,Strategies-for-single-shot-discrimination-of-process-matrices-lewandowska-zbigniew-puchala,Skotiniotis2024,Quantum-change-point-gael-sentis-munoz-tapia}.

The discrimination of quantum states has been the first and most widely studied case, and seminal results such as the Helstrom bound were developed before the 70s \cite{Hellstrom_Bound_helstrom_quantum_1969}. State discrimination opened the door for widespread interest in other variations of the problem as well as providing a large amount of tools and methods for approaching them. A related fundamental task we consider in this work is the discrimination of quantum measurements. While measurements are at the core of quantum theory, and significant progress has been made~\cite{MeasD_Sedl_k_2014,MeasD_Pucha_a_2018,MeasD_PhysRevA.80.052102-unambiguous-ziman-sedlak,MeasD_PhysRevA.90.022317,MeasD_Meas_distance_Ji_2006,MeasD-multiple-shot-Von-Neumann-projective-Zbigniew-Puchala,MeasD_with_post-measurement_Manna_2025,MeasD-rank1-Krawiec-Zbigniew-Puchala,MeasD-labeling-Ragini-Ziman,MeasD-unknown-certification-krawiec-Zbigniew-Puchala,MeasD-Observables-Ziman-Heinosaari}, our understanding of measurement discrimination remains limited when compared to state discrimination.
In particular, apart from the recent work done in~\cite{MeasD_with_post-measurement_Manna_2025}, previous works on measurement discrimination typically focus on POVMs and classical outcomes, without analyzing the role of the associated post-measurement states.

In this work, we investigate the role of post-measurement quantum states in the task of minimum-error measurement discrimination. We introduce the problem of Lüders instrument discrimination, in which the post-measurement state associated with each POVM element is accessible, and formulate it within the framework of quantum channel discrimination. We first prove that in the dichotomic case, for qubit projective measurements, Lüders instrument discrimination is equivalent to the discrimination of two-copy pure states, thereby extending previous results that considered strategies without entanglement~\cite{MeasD_with_post-measurement_Manna_2025} and showing that access to post-measurement states strictly enhances discrimination performance. We further demonstrate that this enhancement can be arbitrarily large: even for qubit measurements, we construct explicit families of POVMs for which the ratio between the optimal discrimination biases with and without post-measurement states diverges. Finally, we complement our analytical results with a semidefinite-programming based numerical investigation of more general scenarios, illustrating the operational relevance of post-measurement states beyond analytically tractable cases.

\section{The measurement discrimination task}
\begin{figure}[ht!]
   \centering
    \includegraphics[width=0.5\textwidth]{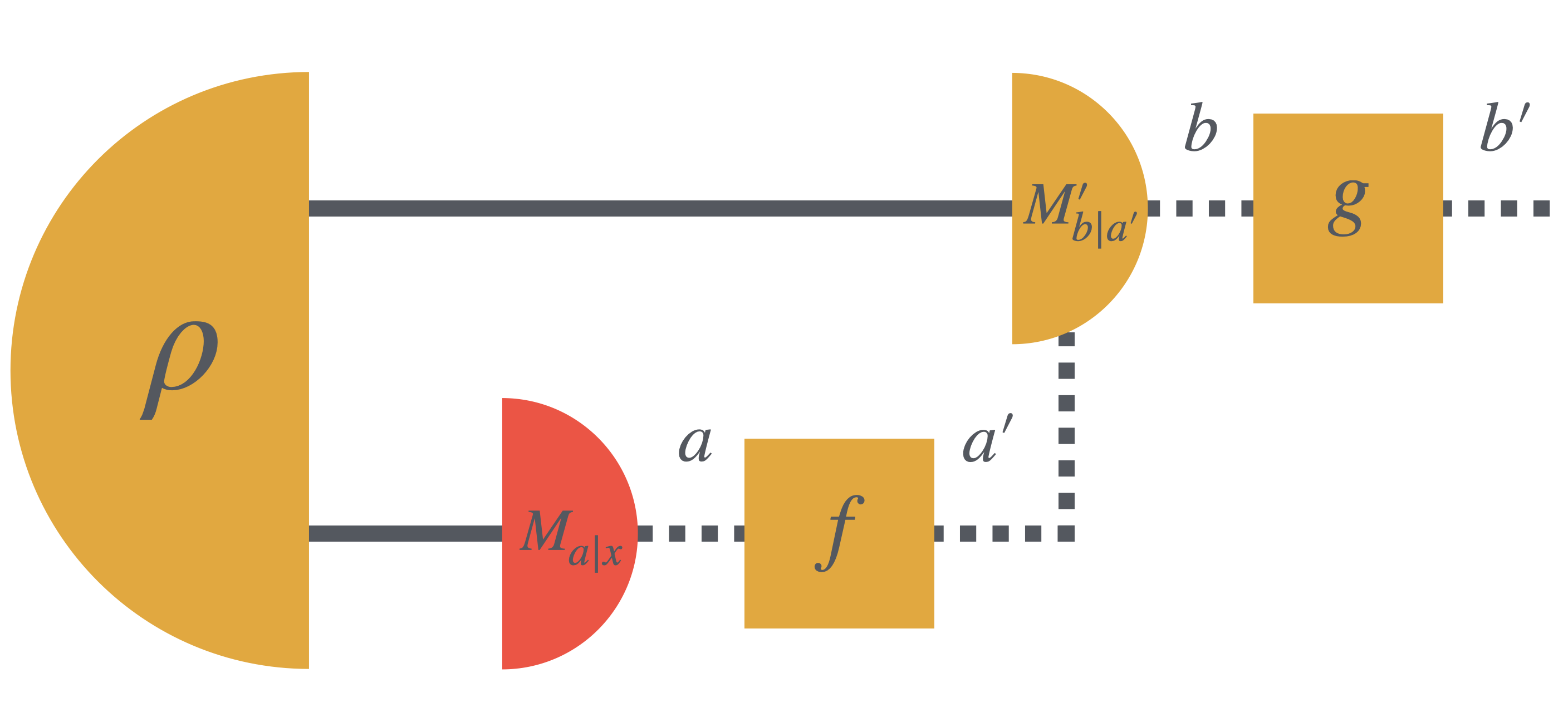}
    \caption{The most general strategy for measurement discrimination without the post-measurement state.
    In yellow: elements of the strategy (free variables to be optimized over). In red: unknown objects being discriminated (fixed parameters that define the discrimination problem). $M_{a|x} \in \Mcal_x$ is the POVM element with outcome $a$. $M'_{b|a'} \in \Mcal'_{a'}$ is the POVM element with outcome $b$. $\rho$ is a bipartite input state. For an outcome $a$ of $\Mcal_x$, $f(a)=a'$ where $f$ is a function that decides which measurement $\Mcal'_{a'}$ is to be applied. For an outcome $b$ of $\Mcal'_{a'}$, $g(b)=b'$ decides the final guess of the strategy.}
    \label{MeasDisc}
\end{figure}
The task of minimum-error one-shot measurement discrimination can be defined as follows \cite{MeasD_Sedl_k_2014, MeasD_Pucha_a_2018, MeasD_PhysRevA.80.052102-unambiguous-ziman-sedlak, MeasD_PhysRevA.90.022317}. With probability $p_x$, Alice is given an unknown quantum measurement $\Mcal_x:=\{M_{a|x}\}_{a=1}^{n_x}$, mathematically described by a positive operator-valued measure (POVM), a set of operators $\{M_{a|x}\}_{a=1}^{n_x}$, which respect $M_{a|x} \geq 0 \,  \forall a,x$ and $\sum_a^{n_x} M_{a|x} = \id \, \forall  x$. In other words, a POVM $\Mcal_x$ is randomly drawn from an ensemble $\Ecal=\{p_x,\Mcal_x\}_{x=1}^N$ which is known to Alice. Being allowed to use the measurement $\Mcal_x$ once, her task is to determine which measurement she received, by inputting a quantum state once into the measurement, observing the outcome and guessing the (classical) value of $x\in\{1,\ldots,N\}$.

The most general strategy for one-shot measurement discrimination (see \cref{MeasDisc}) can be described as follows \cite{MeasD_Sedl_k_2014}: Alice inputs a bipartite state $\rho \in \Lcal(\hil{in}\otimes\hil{aux})$, where $\hil{in}$ is the space over which the measurement acts, and $\hil{aux}$ is an auxiliary space which may be used to improve performance. Then, based on the outcome $a$ of the measurement $\Mcal_x$, she performs a measurement $\Mcal'_{a'}$, chosen from a predefined set of possible measurements $\{\Mcal'_{a'}\}_{a'=1}^{N'}$, over the auxiliary space, and finally, depending on the outcome $b$ of that measurement, she chooses an index $b'$, which will be her guess. If we define $f(a)=a'$ as the function that maps a measurement outcome $a$ of $\Mcal_x$ to a choice of measurement $\Mcal'_{a'}$, $g(b)=b'$, the function that maps a measurement outcome $b$ of $\Mcal'_{a'}$ to a guessed index $b'$, and $\Mcal_x =\{ M_{a|x}\}_{a=1}^{n_x}$, $\Mcal'_{a'} =\{ M'_{b|a'}\}_{b=1}^{n'_{a'}}$, then we can write the probability of Alice correctly guessing the value of $x$, which happens for $b' = x$:
\begin{align}
\begin{split}
    p_{\text{succ}} =& \sum_{x,a,a',b,b'}  p_x \; \delta_{a', f(a)} \delta_{b', g(b)} \delta_{x, b'} \tr{ \rho \left(M_{a|x} \otimes   M'_{b|a'}\right) }
\end{split} \\
=& \sum_b  p_{g(b)} \sum_a \tr{\rho  \left(M_{a|g(b)} \otimes M'_{b|f(a)}\right) }.
\end{align}
As discussed in \cref{methods}, since the functions $f$ and $g$ are simple classical preprocessing, without loss in performance, we may always set $f$ and $g$ as the identity function.

Finally, we note that, while our work focuses on single-shot minimum-error discrimination, other relevant figures of merit for measurement discrimination include unambiguous discrimination~\cite{MeasD_PhysRevA.80.052102-unambiguous-ziman-sedlak,MeasD_Pucha_a_2018}, multiple-shot discrimination~\cite{MeasD-multiple-shot-Von-Neumann-projective-Zbigniew-Puchala}, measurement certification~\cite{MeasD-unknown-certification-krawiec-Zbigniew-Puchala}, and measurement labeling~\cite{MeasD-labeling-Ragini-Ziman}.

\section{Quantum Instruments: measurements with post-measurement states}

In quantum theory, it is not necessarily the case that quantum measurements destroy the quantum system. There are cases where, in addition to a classical outcome, one has access to a post-measurement quantum state. For instance, after performing a projective measurement on some orthonormal basis $\{\phi_a\}_a$, when the outcome $a$ is obtained, one may have access to a post-measurement state given by $\ket{\phi_a}$. More generally, after performing a quantum measurement with the POVM $\Mcal=\{M_a\}_{a=1}^n$ on the state $\rho$, one may have access to classical outcomes $a$ and post-measurement states 
\begin{align} \label{Eq:Luders1}
    \rho_a:= \frac{\sqrt{M_a} \rho \sqrt{M_a}^\dagger} { \tr{M_a \rho}},
\end{align}
where $\sqrt{M_a}$ are the unique positive semidefinite square roots of $M_a$, often called the measurement operators, or the Kraus operators of the quantum measurement ~\cite{Nielsen_Chuang_2010}.

This more general way to define quantum measurements can be formalized by quantum instruments. A quantum instrument \cite{Heinosaari_Ziman_2011, Naimark_Dilation_Modern_Quantum_Instruments_Earliest_POVM_Earliest_davies_operational_1970,Watrous_2018} is a set of trace non-increasing, completely positive (CP) linear maps $\Ical:=\{\map{C}_a\}_{a=1}^{n}$, such that their sum $\sum_{a=1}^n \map{C}_a$ is a quantum channel, that is, a completely positive trace preserving map (CPTP). Quantum instruments represent probabilistic quantum transformations and, thus, are the right object to model the transformation a measurement applies to a state. Given an instrument $\Ical$ and an input state $\rho$, the probability of map $\map{C}_a$ being applied is $p(a|\rho,\Ical) = \tr{\map{C}_a(\rho)}$, with $\sum _a^n\tr{\map{C}_a (\rho)} = 1$ by construction. Each instrument corresponds to a unique POVM, which can be proven to have its elements given by $M_a=\map{C_a}^\dagger(\id)$, where $\map{C_a}^\dagger$ is the adjoint map of $\map{C_a}$. On the other hand, for each measurement $\Mcal$, there are various sets of instruments $\map{C_a}$ such that $\tr{\map{C}_a (\rho)} = \tr{M_a \rho}  \,\, \forall \rho,a$. \cite{Heinosaari_Ziman_2011}

The instrument corresponding to \cref{Eq:Luders1} is known as Lüders' instrument, and it has various physical and mathematical properties \cite{Heinosaari_Ziman_2011}. The Lüders instrument of a POVM $\Mcal =\{M_a\}_{a=1}^n$ is $\Lcal:=\{\map{L}_a\}_{a=1}^{n}$, $\map{L}_a(\rho)= \sqrt{M_a}\rho \sqrt{M_a}$. The probability of a map $\map{L}_a$ being applied here is simply $p(a|\rho,\Lcal)=\tr{\map{L}_a(\rho)} = \tr{M_a \rho}$, correctly reproducing the probability distribution given by the measurement. The quantum channel associated with each of the linear maps $\map{C}_a$ of a quantum instrument $\Ical$ is $\rho_a = \map{C}_a(\rho)/p(a|\rho,\Ical)$ or for Lüders' instruments $\map{L}_a(\rho)/p(a|\rho,\Lcal)$.

Due to their practical and mathematical relevance~\cite{Heinosaari_Ziman_2011}, in this work, we will restrict our analysis to Lüders' instruments. We notice that, if we consider the task of measurement discrimination with arbitrary instruments, one can artificially make this task trivial by choosing some specific instruments. For instance, for any set of POVMs $\{M_{a|x}\}_{a,x}$, we may define the instruments $\map{C}_{a|x}(\rho)=\tr{M_{a|x}\rho}\ketbra{x}{x}$, which performs the measurement $x$ and output ``orthogonal flag-states'' $\ketbra{x}{x}$, which allow us to trivially identify which measurement was performed.

 \begin{figure}[ht!]
    \centering
    \hspace{30pt}
    \includegraphics[width=0.5\textwidth]{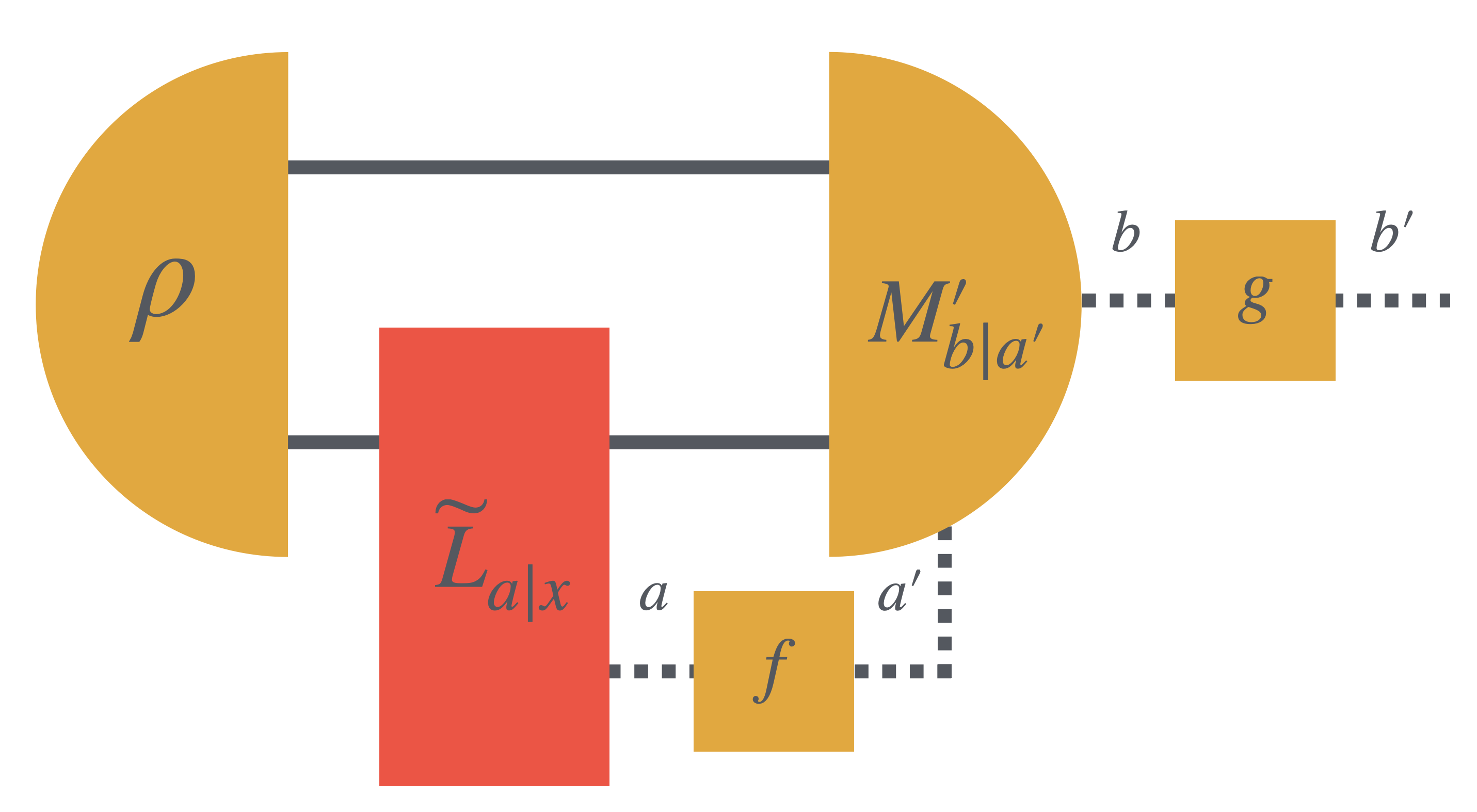}
    \caption{Most general strategy for instrument discrimination in circuit notation. In yellow: elements of the strategy (Free variables to be optimized over). In red: unknown objects being discriminated (fixed parameters that define the discrimination problem). $\map{L}_{a|x} \in \Lcal_x$ is the instrument element with outcome $a$. $M'_{b|a'} \in \Mcal'_{a'}$ is the POVM element with outcome $b$. $\rho$ is a bipartite input state. For an outcome $a$ of $\Lcal_x$, $f(a)=a'$ decides which measurement $\Mcal'_{a'}$ is to be applied. For an outcome $b$ of $\Mcal'_{a'}$, $g(b)=b'$ decides the final guess of the strategy.}
    \label{InstDisc}
\end{figure}

 \section{Instrument discrimination: Measurement discrimination when post-measurement states are available}
We define the problem of minimum-error, one shot Lüders instrument discrimination as follows: with probability $p_i$, Alice is given an unknown Lüders' instrument $\Lcal_x = \{\map{L}_{a|x}\}_{a=1}^{n_x}:\Lcal(\hil{in})\to \Lcal(\hil{out})$ that is drawn from an ensemble $\Ecal=\{p_x,\Lcal_x\}_{x=1}^N$ which is known to her. Being allowed to use the Lüders' instrument $\Lcal_x$ once, her task is to determine which instrument she received, by inputting states into the Lüders' instrument, observing the index of the map $\map{L}_{a|x}$ that has been applied, performing operations on the output state, and guessing the value of $x\in\{1,\ldots,N\}$. 

The most general strategy (see \cref{InstDisc}) now becomes the following: Alice inputs a bipartite state $\rho \in \Lcal(\hil{in}\otimes\hil{aux})$, where $\hil{in}$ is the input space of the Lüders' instrument and $\hil{aux}$ an auxiliary space. Then, based on the classical outcome $a$ of the instrument $\Lcal_x$, she performs a joint measurement $\Mcal'_{a'}$, chosen from a predefined set of possible measurements $\{\Mcal'_{a'}\}_{a'=1}^{N'}$,  over $\hil{out}\otimes\hil{aux}$, the auxiliary and output spaces, and finally, depending on the outcome $b$ of that measurement, she chooses an index $b'$, which will be her final guess. If again, we define $f(a)=a'$ as the function that maps a classical outcome $a$ of $I_x$ to a choice of measurement $\Mcal'_{a'}$, and $g(b)=b'$, the function that maps a measurement outcome $b$ of $\Mcal'_{a'}=\{ M'_{b|a'}\}_{b=1}^{n'_{a'}}$ to a guessed index $b'$, then we can write the probability of Alice correctly guessing the value of $x$, which is a sum over the probabilities of the events where $b' = x$:
 \begin{align}
 \begin{split}
      p_{\text{succ}} = \sum &p_x \delta_{a',f(a)} \delta_{b',g(b)} \delta_{x,b'} \\ & \tr{(\map{L}_{a|x} \otimes \id_{\text{aux}})(\rho) M'_{b|a'}}.
 \end{split}
 \end{align}
As in the measurement discrimination case, and as discussed in \cref{methods}, since the functions $f$ and $g$ are simple classical preprocessing, without loss in performance, we may always set $f$ and $g$ as the identity function.
 
For any ensemble of POVMs $\Ecal = \{p_x,\Mcal_x\}_{x=1}^N$, one can always achieve at least the same success probability for the Lüders instrument discrimination task as for the measurement discrimination task by applying the same strategy on the classical outcomes and auxiliary spaces ignoring the quantum output space. The main question of this work is to check when the post-measurement state is useful, and how useful it is. 

\section{Methods: quantum testers and channel discrimination}
\label{methods} 
In order to simplify our analysis, it is convenient to represent quantum measurements and quantum instruments as quantum channels, as done in \cite{MeasD_Sedl_k_2014}. In this way, the problem of quantum instrument discrimination is reduced to a quantum channel discrimination task. The problem of minimum-error one shot channel discrimination is the following \cite{Watrous_2018}: Alice is given an unknown quantum channel (CPTP map) $\map{C}_x:\Lcal(\hil{in})\to\Lcal(\hil{out})$, drawn from an ensemble $\Ecal=\{p_x,\map{C}_x\}_{x=1}^N$ that is known to her. After being allowed to use the channel $\map{C}_x$ once, her task is to determine which channel she received, by performing operations on it and guessing the value of $x\in\{1,\ldots,N\}$.

In order to analyze discrimination between quantum channels, it is convenient to represent channels via the Choi-Jamiołkowski isomorphism \cite{Choi_operators_CHOI1975285,Choi_operators_Jamiolkowski:1972pzh}. The Choi operator of a linear map $\map{C}:\Lcal(\hil{in})\to\Lcal(\hil{out}) $ is defined as $C  := (\id \otimes \map{C})(\kebra{\Phi^+}) \in \Lcal(\hil{in} \otimes \hil{out})$, where $\ket{\Phi^+} = \sum_{j=0}^{d-1} \ket{j} \otimes \ket{j}$ is a unnormalized maximally entangled state. A linear map $\map{C}$ is completely positive (CP) if and only if $C\geq0$, and trace preserving (TP) if and only if $\Tr_{\text{out}}(C) = \id_{\text{in}}$. We now thus have a one-to-one mapping between linear maps and Choi operators as well as strict conditions for identifying Choi operators that correspond to quantum channels. 

Being able to now treat quantum channels as matrices, we introduce process POVMs (PPOVMs)~\cite{Testers_Ziman_2008}, also known as quantum testers \cite{Testers_Link_Product_Chiribella_2009}. Testers are the analogous of POVMs to quantum channels, and may be viewed as a higher-order version of POVMs \cite{Higher_Order_taranto2025higherorderquantumoperations}, since they act as measurement devices for quantum channels. A tester is a set of linear operators $\Tcal:=\{T_b\}_{b=1}^N$, $T_b \in \Lcal(\hil{in} \otimes \hil{out})$, satisfying,  $T_b \geq0\,\forall b$ and $\sum_b^N T_b = \sigma \otimes \id_{\normalfont{out}}$, where $\sigma \in \Lcal(\hil{in})$ is a quantum state.  It holds that $p(b|\map{C},\Tcal) = \Tr(T_b C)$ is the probability of obtaining outcome $b$ for a channel discrimination task where we use a strategy with corresponding tester $\Tcal$, knowing the correct channel is $\map{C}$. Moreover, one can prove that every tester $\Tcal$ admits a physical realization as an admissible channel discrimination strategy, as well as that every possible strategy admits a valid tester, making the study of channel discrimination extensively analyzable through the lens of testers without loss of either realizability or attainability \cite{Testers_Ziman_2008}.

\begin{figure}[ht!]
   \centering
    \includegraphics[width=0.5\textwidth]{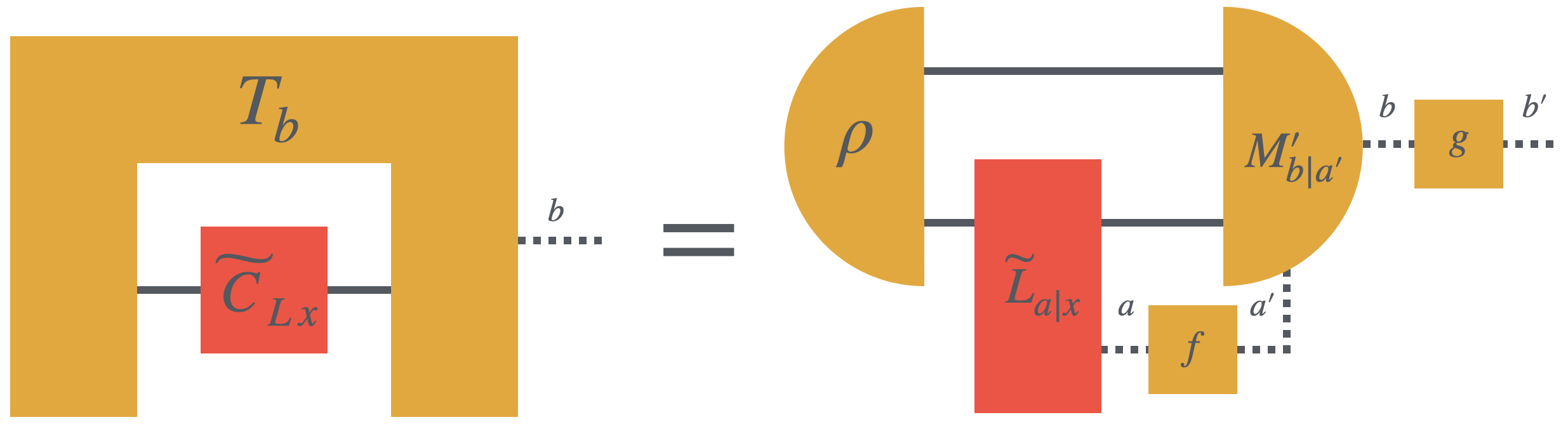}
    \caption{Instrument discrimination as a channel discrimination task.
    In yellow: elements of the strategy (free variables to be optimized over). In red: unknown objects being discriminated (Fixed parameters that define the discrimination problem). To the right of the reader, the most general instrument discrimination strategy, to the left the equivalent channel discrimination strategy that uses testers.}
    \label{ChanDisc}
\end{figure}

To address the problem of Lüders' instrument discrimination using this formalism, we can express the Lüders instrument $\Lcal =\{\map{L}_a\}_{a=1}^n$ of a measurement $\Mcal =\{M_a\}_{a=1}^n$ as a channel $\map {C_L(}\rho): \Lcal(\hil{in}) \to \Lcal(\hil{out1} \otimes \hil{out2})$, which for an input $\rho$ outputs the equivalent of a classical index $a$ on the second space, by preparing the state $\kebra a$ (measure-and-prepare), and the post-measurement state on the first space \cite{Higher_Order_taranto2025higherorderquantumoperations}:

\begin{align}
    \map {C_L(}\rho) :=& \sum_{a=1}^n \map{L}_a(\rho) \otimes \kebra{a}
    \\ =&\sum_{a=1}^n \sqrt{M_a}\rho \sqrt{M_a} \otimes \kebra{a}.
    \label{ludchan}
\end{align}

Under the assumption that all instruments being discriminated have the same number of outcomes, the task of discriminating an ensemble $\Ecal=\{p_x,\Lcal_x\}_{x=1}^N$ of instruments is equivalent to the discrimination of their respective channels $\Ecal'=\{p_x,\map{C_L}_x\}_{x=1}^N$. The most general strategy then simply becomes to apply a tester $\Tcal$ to the given channel and have the outcome of the tester be our guess. One could argue that we could also use a function to post process the classical output of the tester, but for any tester $\Tcal$ followed by a function $f$, there always exists a tester $\Tcal'$ and a function $f'$ that output the same guesses, making the post processing redundant. Denoting with ${C_L}_x$ the Choi operator of the channel $\map {C_L}_x(\rho)$, the success probability of the instrument discrimination task then becomes:
\begin{align}
    p_{\text{succ}} = \sum_{x=1}^N p_x p(x|\map{C_L}_x, \Tcal) =\sum_{x=1}^N p_x \tr{{C_L}_x T_x}
\end{align}
This form of the problem is particularly interesting, because the function for $p_{\text{succ}}$ has now become linear in $T_x$ and can be numerically maximized via semidefinite programming   \cite{ChannelD_Bavarescooo_2021, 2022JMP....63d2203B}.

It is noteworthy that in the dichotomic discrimination case, the channel discrimination formulation of Lüders instrument discrimination can be used to define an alternate notion of distance between quantum measurements which does take into account the post measurement state.  Let $\map{C_L}_1$ and $\map{C_L}_2$ be the quantum channels of the Lüders instruments corresponding to measurements $\Mcal_1$ and $\Mcal_2$ as defined in \cref{ludchan}. We thus define the Lüders distance between two POVMs $\Mcal_1$ and $\Mcal_2$:
\begin{align}
\label{ludist}
    d_L(\Mcal_1,\Mcal_2) := \norm{\map{C_L}_2 - \map{C_L}_1}_\diamond
\end{align}
where $\norm{.}_\diamond$ is the diamond norm (completely bounded trace norm) \cite{Diamond_norm_aharonov1998quantumcircuitsmixedstates,Watrous_2018}, defined as 
\begin{align}
    \norm{\map{C}}_\diamond := \max_{\rho_\text{in,aux}} \norm{\map{C}\otimes \map{\id_\text{aux}}(\rho_{\text{in,aux}})}_1, 
\end{align}
where $\rho\in \mathcal{L}(\mathcal{H}_\text{in}\otimes \mathcal{H}_\text{aux})$ is a quantum state and $\mathcal{H}_\text{aux}$ is an arbitrary auxiliary linear space, and $\norm{A}_1:\Tr(\sqrt{A A^\dagger})$ is the one norm. Observe that the Lüders distance between two POVMs $d_L$ shares the same properties of the diamond norm for channels and it is indeed a mathematical distance (metric). Since the optimal success probability $p_C$ for discriminating two quantum channels (with equal prior probabilities) $\map{C_1}$ and $\map{C_2}$ is $p_C=\frac{1}{2} + \frac{1}{4} \norm{\map{C_2}-\map{C_1}}_\diamond$ \cite{Watrous_2018}, we may then use the above definition to write the optimal discrimination success probability for two measurements with access to the post-measurement state as:
\begin{align}
    p_{\text{succ}} = \frac{1}{2}+ \frac{d_L(\Mcal_1,\Mcal_2)}{4}.
\end{align}

We will also note $d_M(\Mcal_1,\Mcal_2)$, the distance between two measurements $\Mcal_1$ and $\Mcal_2$ defined in an analogous way. Consider a POVM $\Mcal=\{M_a\}_{a=1}^n$ with the measure-and-prepare channel channel
\begin{align}
\map{\Mcal}(\rho)= \sum_{a=1}^n \tr{M_a \rho} \kebra a.    
\end{align}
We now define the distance between two measurements $\Mcal_1$ and $\Mcal_2$ (without the post-measurement state) as
\begin{equation} \label{eq:d_M}
    d_M(\Mcal_1,\Mcal_2):=\norm{\map{M}_2 - \map{M}_1}_\diamond,
\end{equation} then, the optimal discrimination success probability for two measurements without access to the post-measurement state is
    $p_{\text{succ}} = \frac{1}{2}+ \frac{d_M(\Mcal_1,\Mcal_2)}{4}$.

\section{Main result 1: analytic solution for projective qubit instruments and optimal strategies}

\subsection{Success probability}

We now present our first main result which considers the simple instance of a pair of projective qubit Lüders instruments. We give a closed formula which takes the from of a Helstrom bound for the discrimination of two-copy pure states.

\begin{theorem} \label{thm:main1}
    In the context of one-shot minimum-error discrimination, the task of discriminating two projective qubit Lüders instruments characterized by POVMs $\{\kebra{\psi}, \kebra{\psi^\perp} \}$\footnote{Here, $\ket{\psi^\perp}$ is any quantum state which is orthogonal to $\ket{\psi}$, hence, we have that $\ketbra{\psi^\perp}=\id - \ketbra{\psi}$.} and $\{\kebra{\phi}, \kebra{\phi^\perp} \}$, and occurring with respective probabilities $p_{\psi}$ and $p_{\phi}$, is mathematically equivalent to the discrimination of the pure states $\kebra{\psi}^{\otimes2}$ and $\kebra{\phi}^{\otimes2}$ occurring with the same respective probabilities. Hence, the optimal success probability is given by $ \max p_s =\frac{1}{2}\left(1+\sqrt{1-4p_{\psi}p_{\phi}|\langle\psi \psi|\phi \phi\rangle|^2}\right)$.

    Moreover, the maximal performance is attained without the use of entanglement.
\end{theorem}
We note that \cref{thm:main1} extends the results from \cite{MeasD_with_post-measurement_Manna_2025}. More precisely, Thm. 2 from \cite{MeasD_with_post-measurement_Manna_2025} considers the task of discriminating two dichotomic qubit Lüders instruments without entanglement, and shows that the optimal performance is given by $\max p_s =\frac{1}{2}\left(1+\sqrt{1-4p_{\psi}p_{\phi}|\langle\psi \psi|\phi \phi\rangle|^2}\right)$. Here, we show that entanglement cannot be used to improve performance. Our proof follows different steps from \cite{MeasD_with_post-measurement_Manna_2025}, of which we only became aware after the completion of this work.

Using \cref{thm:main1}, we see that, apart from the case where the measurements can be perfectly discriminated, having access to the post measurement state always improves the optimal success probability compared to the case where we do not use the post-measurement state. Indeed, the optimal discrimination of a pair of projective qubit measurements, as shown in \cite{MeasD_Sedl_k_2014}, is equivalent to one copy pure state discrimination and has success probability $ \max p_s' =\frac{1}{2}\left(1+\sqrt{1-4p_{\psi}p_{\phi}|\langle\psi |\phi \rangle|^2}\right)$. 

Take for instance the case of discriminating the projective measurements $\Zcal = \{\kebra0,\kebra1\}$ and $\Xcal =\{\kebra+,\kebra-\}$, occurring with equal probabilities, and where $\ket+ =\frac{1}{\sqrt2}\left(\ket0+\ket1\right)$. With no access to the post-measurement state, this is equivalent to the discrimination of $\ket0$ and $\ket+$ and has maximum success probability $p=\frac{1}{2}\left(1+\sqrt{1-|\bra0 +\rangle|^2}\right) = \cos^2(\pi/8) \approx 0.853$, while with access to the post-measurement state, the success probability becomes $p'=\frac{1}{2}\left(1+\sqrt{1-|\bra{00} ++\rangle|^2}\right)  \approx 0.933$, which constitutes a significant increase.

Below, we present the core of the proof of \cref{thm:main1}. The proof that optimal performance is attainable without entanglement is presented in \cref{subsec:without_ent}.

\begin{proof}
Consider two projective qubit measurements $\Mcal_1 = \{M_{1|1},M_{2|1}\} =\{\kebra{\psi},\kebra{\psi^\perp}\}$ and $\Mcal_2 = \{M_{1|2},M_{2|2}\} =\{\kebra{\phi},\kebra{\phi^{\perp}}\}$, and the task of optimally discriminating their respective Lüders' Instruments $\Lcal_1$ and $\Lcal_2$, occurring with probabilities $p_1$ and $p_2$. Using a tester $\Tcal$ to accomplish this task, we can write the probability of successful discrimination:

\begin{align}
    p_{\text{succ}} &= p_1 p(1|\map{C_L}_1,\Tcal) + p_2 p(2|\map{C_L}_2,\Tcal) \\&= p_1 \tr{{C_L}_1 T_1} +p_2\tr{{C_L}_2 T_2},
\end{align}

where $\map{C_L}_x:\Lcal(\hil{in}) \rightarrow \Lcal(\hil{out1}\otimes \hil{out2}) $, ${C_L}_x \in \Lcal(\hil{in}\otimes \hil{out1}\otimes \hil{out2})$ are the corresponding channels and their Choi operators:
\begin{align}
    {C_L}_x = \sum_{a=1}^2 M_{a|x}^T \otimes M_{a|x} \otimes \kebra{a}
\end{align}
where $.^T$ stands for transposition in the computational basis. 

Let us introduce the projectors $\pi_a = \id \otimes \id \otimes \kebra{a}$, for which we have the identity $\sum_a \pi_a{C_L}_x\pi_a = {C_L}_x$. From this follows that:
\begin{align}
    p(b,\map{C_L}_x|\Tcal) &= \tr{{C_L}_x T_b}  = \sum_a \tr{\pi_a{C_L}_x \pi_a T_b }\\ &= \sum_a \tr{{C_L}_x \pi_a T_b \pi_a}. 
\end{align}
In this way, for all $\map{C_L}_x$, given a tester $\Tcal \ne \pi(\Tcal)$,  we can always use $\pi(\Tcal)$ instead, while preserving all conditional probabilities, and since $\pi(\pi(\Tcal)) = \pi(\Tcal)$, all testers we now consider satisfy the relation $\Tcal =\pi(\Tcal)$. The most general form for $\Tcal$ we will be considering is thus:
\begin{align}
    T_b = \sum_a \sigma_{a|b} \otimes \kebra{a},
\end{align}

with operators $\sigma_{a|b} \in \Lcal(\hil{in}\otimes\hil{out1})$, $\sigma_{a|b}\geq 0$ and, following from the normalization condition on $\{T_b\}_b$, $\forall a$, we have:
\begin{align}
    \sum_b \sigma_{a|b} = \rho \otimes \id ,
\end{align}
where $\rho \in \Lcal(\hil{in})$ is a quantum state.

Furthermore, we introduce the transformations $\Gamma \otimes \Gamma \otimes \hat{X} (.)$ where $\hat{X}(.)$ is conjugation by the Pauli X operator and $\Gamma(A) = \frac{\tr{A}\id-A}{\dim(A)-1}$ is the universal NOT transformation, which is a positive and trace-preserving linear map~\cite{Buzek1999NOT,Rungta2001NOT}. For the dimension of $A$, $\dim(A)=2$, we have:
\begin{align}
    T'_b=\Gamma \otimes \Gamma \otimes \hat{X} (T_b) &= \sum_a \Gamma\otimes\Gamma (\sigma_{a|b}) \otimes \hat{X}\kebra{a}\hat{X} \\&= \sum_a \Gamma\otimes \Gamma(\sigma_{a+1|b})\otimes \kebra{a},
\end{align}
then:
\begin{align}
    p&(b,\map{C_L}_x|\Tcal')  = \text{Tr}({C_L}_xT'_b)  \\
    &=\sum_a \text{Tr}(\Gamma\otimes\Gamma (\sigma_{a+1|b})(M_{a|x}^T \otimes M_{a|x})) \\
    &=\sum_a \text{Tr}( (\sigma_{a+1|b})\Gamma\otimes\Gamma(M_{a|x}^T \otimes M_{a|x})) \\
    &=\sum_a \text{Tr}(\sigma_{a+1|b}(M_{a+1|x}^T \otimes M_{a+1|x})) \\
    &= p(b,\map{C_L}_x|\Tcal) ,
\end{align}
where we used the properties $\text{Tr}(\Gamma \otimes \Gamma (A)B) =\text{Tr}(\Gamma \otimes \Gamma(B)A)$ and $\Gamma(M_{a|x}) = M_{a+1|x}$. Using the same argument as earlier, for any tester $\Tcal$ that does not satisfy $\Tcal = \Gamma \otimes \Gamma \otimes \hat{X} (\Tcal)$, we can use the tester $\Gamma \otimes \Gamma \otimes \hat{X} (\Tcal)$ instead, which does while preserving all conditional probabilities. Thus we limit the set of testers we optimize over to ones that satisfy this property and therefore also meet:
\begin{align}
\label{noteq}
    \Gamma \otimes \Gamma (\sigma_{a+1|b}) = \sigma_{a|b}
\end{align}

Using this last equation, we can further restrict the normalization of $\{\sigma_{a|b}\}_b$, since:
\begin{align}
\label{equ: norm2}
    \sum_b \sigma_{a|b} &= \sum_b \Gamma\otimes\Gamma(\sigma_{a+1|b}) \\&=\rho \otimes \id = \Gamma \otimes \Gamma (\rho \otimes \id ) \\&= \frac{1}{2} \id\otimes\id,
\end{align}
given that $\rho=\frac{1}{2} \id$ is the only solution to $\rho = \Gamma(\rho)$. With the properties that $\Gamma \otimes \Gamma$ in $\dim(A)=2$ is self-adjoint and self-inverse, we can show that:
\begin{align}
   & p(b,\map{C_L}_x|\Tcal) \\&= \text{Tr}( (M_{1|x}^T\otimes M_{1|x})\sigma_{1|b} + (M_{2|x}^T\otimes M_{2|x})\sigma_{2|b}) \\&= \text{Tr}( (M_{1|x}^T\otimes M_{1|x})\sigma_{1|b} + (M_{1|x}^T\otimes M_{1|x}) \Gamma(\sigma_{2|b})) \\& = \text{Tr}(2\sigma_{1|b}(M_{1|x}^T\otimes M_{1|x})) \\&= \text{Tr}(N_b(M_{1|x}^T\otimes M_{1|x})) ,
\end{align}
where the operators $N_b = 2\sigma_{1|b}$ obey $N_b \geq 0$ and $\sum_b N_b = \id$ form by definition a POVM. And, for the discrimination of two projective qubit instruments, we now have:
\begin{equation}
\begin{split}
     p_{\text{succ}}& = p_1 \tr{N_1\kebra{\psi} \otimes \kebra{\psi}^T} \\&+ p_2 \tr{N_2\kebra{\phi} \otimes \kebra{\phi}^T}
\end{split},
\end{equation}
where $p_{\text{succ}}$ is the success probability of discriminating states $\kebra{\psi} \otimes \kebra{\psi}^T$ and $\kebra{\phi} \otimes \kebra{\phi}^T$ occurring with respective probabilities $p_1$, $p_2$, and using POVM $\{N_b\}_b$. For pure qubits, as is the case here, this is in turn equivalent to the discrimination of states $\kebra{\psi}^{\otimes 2}$ and  $\kebra{\phi}^{\otimes 2}$. 
\end{proof}

\subsection{Attaining maximal performance with entanglement} \label{subsec:with_ent}

In this subsection, we show how to attain the optimal performance presented in \cref{thm:main1} with entangled states, and in the following subsection, we show how entanglement is not required to attain the optimal performance.

An entanglement based strategy can be derived using elements from the proof of the theorem. This is done by first writing the normalization condition given in \cref{equ: norm2} as a normalization of the tester elements:

\begin{align}
    \sum_b T_b =\sum_{a,b} \sigma_{a|b} \otimes \kebra{a} = \frac{1}{2} \id\otimes\id\otimes\id.
\end{align}

We then use the equation \cite{Testers_Ziman_2008,Testers_Link_Product_Chiribella_2009}:
\begin{align}
    \sum_b T_b = \text{Tr}_\text{aux}\left(\rho\right)\otimes\id\otimes\id,
\end{align}
where $\rho$ is the bipartite input state and $\text{Tr}_\text{aux}$ is the partial trace over the auxiliary space, to write:
\begin{align}
\label{equ: ent in state}
    \text{Tr}_\text{aux}\left(\rho\right) = \frac{1}{2} \id
\end{align}
To reiterate the argument (see \cref{thm:main1}): For any tester that does not satisfy the above conditions, there exists one that achieves the same success probability, yet does. We can thus use \cref{equ: ent in state}, to find the input state $\rho$ of the optimal tester, which has to be a maximally entangled state as its partial trace is $\frac{1}{2} \id$. If one inputs the maximally entangled state$|\phi_2^+\rangle = \frac{1}{\sqrt2}\left(\ket{00}+\ket{11}\right)$, into a Lüders instrument (see \cref{InstDisc}), then for a measurement element $\kebra\psi$, direct calculation shows that we get the unnormalized total output state:
\begin{align} \label{eq:psi_psi_conj}
    \left(\kebra{\psi}\otimes\id\right)|\phi_2^+\rangle = \frac{1}{\sqrt{2}} \ket{\psi} \otimes \overline{\ket{\psi}},
\end{align}
where $\overline{\ket{\psi}}$ is the complex conjugate of $\ket{\psi}$ in the computational basis. One way to prove the identity of \cref{eq:psi_psi_conj} is to use the fact that for any $A:\mathcal{H}_{\text{in}} \to \mathcal{H}_{\text{out}}$, we have that
\begin{align}
    A \otimes \id_\text{in} \ket{\phi^+_{d_\text{in}}} = \sqrt{\frac{d_\text{in}}{d_\text{out}}} \id_\text{out} \otimes A^T \ket{\phi^+_{d_\text{out}}}.
\end{align}
In this way, if we view $\kebra{\psi}$ as the composition of the operator $\ket{\psi}:\mathbb{C}^d\to \mathbb{C}^1$ with $\bra{\psi}:\mathbb{C}^1\to \mathbb{C}^d$, and note that $\bra{\psi}^T=\overline{\ket{\psi}}$, we obtain \cref{eq:psi_psi_conj}.

At this point, we recognise that the problem of discriminating Lüders instruments characterized by POVMs $\{\kebra{\psi}, \kebra{\psi^\perp} \}$ and $\{\kebra{\phi}, \kebra{\phi^\perp} \}$ is equivalent to discriminating the states $\ket{\psi} \overline{\ket{\psi}}$ from $\ket{\psi^\perp} \overline{    \ket{\psi^\perp}}$. Now, we notice that the operator  
$U:=\overline{\ket{\psi}} \bra{\psi} + \overline{\ket{\psi^\perp}} \bra{\psi^\perp}$ is unitary. Hence, the pair of states $\{\ket{\psi}, \ket{\psi^\perp} \} $ is related to the  pair of states $\{\overline{\ket{\psi}}, \overline{\ket{\psi^\perp}} \} $ via the unitary operator $U$, hence both sets of states are unitarily equivalent, and discriminating $\ket{\psi} \overline{\ket{\psi}}$ from $\ket{\psi^\perp} \overline{\ket{\psi^\perp}}$ is equivalent to discriminating $\ket{\psi} {\ket{\psi}}$ from $\ket{\psi^\perp} \ket{{\psi^\perp}}$.

\subsection{Attaining maximal performance without entanglement} \label{subsec:without_ent}

We now show how to attain the maximum discrimination success probability without the use of entanglement. For that, we use a flipping symmetry argument similarly to \cref{thm:main1} to show that in the dichotomic case, `projective qubit instruments discrimination with no access to entanglement' is equivalent to `two-copy pure state discrimination with no access to entangling measurements'. Then, use the fact that optimal two-copy pure state discrimination is attainable without access to entangling measurements \cite{StateD-multiple-copy-state-discrimination-Ac_n_2005}. 

The most general strategy that uses no entanglement will simply be to input a state $\rho \in \Lcal(\hil{in})$ into the Lüders instrument $\Lcal_x$, then observe the classical outcome $a$ and apply a corresponding measurement $\Mcal'_a=\{M_{b|a}\}_{b=1}^{n_a'}$ to the post-measurement state. The maximum success probability for the discrimination of two Lüders instruments generated by two projective qubit measurements $\Mcal_{x\in \{0,1\}} = \{\kebra\psi_{a|x}\}_{a,x\in\{0,1\}}$, with prior probabilities $p_{a\in\{0,1\}}$ can be written as follows:

\begin{equation}
\begin{split}
   \max_{\rho,\{M_{b|a}\}} &p_s=\\
   \max_{\rho,\{M_{b|a}\}}
    &p_0\tr{\rho \kebra{\psi_{0|0}} } \tr{\kebra{\psi_{0|0}} M'_{0|0}} \\ 
    +&p_0\tr{\rho \kebra{\psi_{1|0}} } \tr{\kebra{\psi_{1|0}} M'_{0|1}}\\
    +&p_1\tr{\rho \kebra{\psi_{0|1}} } \tr{\kebra{\psi_{0|1}} M'_{1|0}}\\
    +&p_1\tr{\rho \kebra{\psi_{1|1}} } \tr{\kebra{\psi_{1|1}} M'_{1|1}}\\
\end{split}
\end{equation}
Using the property $\tr{\Gamma (A)B} = \tr{A\,\Gamma (B)}$ for $A,B$ 2 dimensional projectors, we can then rewrite the above expression:

\begin{equation}
\begin{split}
   \max_{\rho,\{M_{b|a}\}} &p_s=\\
   \max_{\rho,\{M_{b|a}\}}
    &p_0\tr{\rho \kebra{\psi_{0|0}} } \tr{\kebra{\psi_{0|0}} M'_{0|0}} \\ 
    +&p_0\tr{\Gamma(\rho) \kebra{\psi_{0|0}} } \tr{\kebra{\psi_{0|0}} M'_{1|1}}\\
    +&p_1\tr{\rho \kebra{\psi_{0|1}} } \tr{\kebra{\psi_{0|1}} M'_{1|0}}\\
    +&p_1\tr{\Gamma(\rho) \kebra{\psi_{0|1}} } \tr{\kebra{\psi_{0|1}} M'_{0|1}}\\
\end{split}
\end{equation}

One can then verify that the above expression is the same as the maximum success probability of sequentially discriminating pure states $\ket{\psi_{0|0}\psi_{0|0}}$ and $\ket{\psi_{0|1}\psi_{0|1}}$, with prior probabilities $p_0,p_1$, by sequentially using measurement $\{\rho,\Gamma(\rho)\}$ on the first copy then $\{M'_{0|0}, M'_{1|0}\}$ or $\{M'_{1|1}, M'_{0|1}\}$ on the second copy. (Notice here that we have flipped the indices of the measurement $\Mcal'_1$.)

The strategy that achieves this success probability, for the equal distribution case, will be to first input the optimal state for the discrimination of the two measurements without access to the post-measurement state and without entanglement. Explicitly, the pure state whose Bloch vector is the bisector of the two vectors corresponding to states $\ket{\psi_{0|0}}$ and $\ket{\psi_{1|1}}$. If the output is $0$, then apply the optimal measurement for discriminating $\ket{\psi_{0|0}}$ and $\ket{\psi_{0|1}}$ occurring with respective probabilities $p_M$ and $1-p_M$. If the the output is $1$, then apply the optimal measurement for discriminating $\ket{\psi_{1|1}}$ and $\ket{\psi_{1|0}}$ occurring with respective probabilities $p_M$ and $1-p_M$. We call $p_M$, the optimal success probability of discriminating measurements $\Mcal_0$ and $\Mcal_1$, with no access to the post-measurement state, as given in \cite{MeasD_Sedl_k_2014}.

\section{Main result 2: Large advantage of using the post-measurement state }
\label{UnboundedAdvantageForCertainTasks}

As proven in the previous section, having access to a post-measurement state allows us to outperform the strategies that do not take it into account, for the special case of projective qubit measurements. A natural question to ask is then, how large can this advantage be in the general case? In order to compare the performance of strategies that do consider the post measurement state to those that do not, we define the \textit{instrument advantage bias} as the ratio of the bias between the two optimal strategies' success probabilities.
\begin{align} \label{eq:delta_bias}
    \Delta(\Mcal_1 , \Mcal_2) :=& \frac{p_{s,\max}^{\Lcal}(\Mcal_1, \Mcal_2) - 0.5}{p_{s,\max}^{\Mcal}(\Mcal_1, \Mcal_2) - 0.5} \\
    =&\frac{d_L(\Mcal_1 , \Mcal_2)}{d_M(\Mcal_1 , \Mcal_2)},
\end{align}
where:
\begin{itemize}
    \item \( p_{s,\max}^{\Mcal}(\Mcal_1, \Mcal_2) \) denotes the optimal success probability for discriminating between the POVMs \( \Mcal_1 \) and \( \Mcal_2 \), occurring with equal probabilities.
    
    \item \( p_{s,\max}^{\Lcal}(\Mcal_1, \Mcal_2) \) denotes the optimal success probability for discriminating between the Lüders instruments \( \Lcal_1 \) and \( \Lcal_2 \), occurring with equal probabilities.
\end{itemize}
and  $d_L(.,.)$ and $d_M(.,.)$ are respectively the Lüders' distance and the post-measurement agnostic distance as defined in \cref{ludist} and \cref{eq:d_M}.

Using \cref{thm:main1}, we see that for a pair of projective qubit measurements occurring with equal probabilities, the advantage can be written:
\begin{align}
    \Delta(\Mcal_\psi , \Mcal_\phi) &= \frac{\sqrt{1-|\langle \psi |\phi \rangle|^4}}{\sqrt{1-|\langle \psi |\phi \rangle|^2}} \\ 
    &= \sqrt{1+|\langle \psi |\phi \rangle|^2} . 
\end{align}
Hence, we see that, $\Delta(\Mcal_\psi , \Mcal_\phi)\leq \sqrt{2}$, is an upper bound which is attained in the limit where $\braket{\psi}{\phi}=1$, that is, when the two measurements are the same.

In the following, we  prove that, when the pair of discriminated measurements are not restricted to sets of two rank-1 projectors, the instrument advantage bias can get arbitrarily large. 

\begin{theorem}
Even when considering qubit measurements, the instrument advantage bias (see~\cref{eq:delta_bias}) \( \Delta(\Mcal_\psi , \Mcal_\phi)  \) can be arbitrarily large. That is, for any positive real number \( L\in\mathbb{R}  \), there exist POVMs \(\Mcal_\psi, \Mcal_\phi\) such that \( \Delta(\Mcal_\psi , \Mcal_\phi)> L \).
\label{theo2}
\end{theorem}

The proof of \cref{theo2} is constructive. For that, we consider the POVM $\Zcal=\{\kebra0 ,\kebra1\}$, and  the family of POVMs
\begin{align}
\Wcal^p:= \{\kebra0 +p\kebra1 ,(1-p)\kebra1\},
\end{align}
where $p\in[0,1]$. The POVM $W^p$ can be seen, as a noisy version of $\Zcal$, where for $p=0$, $W^{p=0}=\Zcal$, and when $p=1$, $W^{p=1}$ is a measurement that outputs $0$ regardless of the input state. We will prove that the instrument advantage bias $\Delta(\Zcal,\Wcal^p)$ gets arbitrarily large as $p\rightarrow 0$. 
It is also noteworthy that all POVM elements of the pair of measurements $\Zcal$ and $\Wcal^p$ are diagonal in the computational basis, hence they may be viewed as classical measurements. This shows that post-measurement states also enables large advantages in the classical regime.

\begin{proof}
Let us denote $\Wcal^p:= \{\kebra0 +p\kebra1 ,(1-p)\kebra1 \}$, which defines a measurement for all $p \in [0,1]$. Consider now the discrimination of $\Wcal^p$ and $\Zcal$ occurring with equal probabilities. We can derive the optimal probability of discriminating these measurement without access to the post-measurement state by first writing each of the measurements as a quantum channel and then solving for the corresponding channel discrimination task. For a measurement $\Mcal = \{M_a\}_{a=1}^n$, we will represent the corresponding quantum channel as follows:
\begin{align}
    \map{M}(\rho) = \sum_{a=1}^n \tr{M_a \rho} \kebra{a}
\end{align}

We will write the success probability of the channel discrimination strategy where we input state $\rho \in \Lcal\left(\hil{aux} \otimes \hil{in}\right)$, $\rho = \sum_{ijkl} \rho_{ikjl} \ketbra{i}{j}\otimes \ketbra{k}{l}$ and apply the optimal measurement for the given input state at the end, as the success probability of optimally discriminating states $\id \otimes \map{\Zcal} \left(\rho\right)$ and $\id \otimes \map{\Wcal^p} \left(\rho\right)$. We then have: 
\begin{equation}
    p_{s,\rho}^{\mathcal{M}}(\Zcal, \Wcal^p) = \frac{1}{2} + \frac{1}{4} \norm{(  \id\otimes\map{\Wcal^p})(\rho)-  (\id\otimes  \map{\Zcal})  (\rho) }_1,
\end{equation}

where $\norm{.}_1$ represents the trace norm or 1-norm. After replacing all terms with their corresponding expressions, the equation simplifies to:
\begin{align}
    p_{s,\rho}^{\mathcal{M}}(\Zcal, \Wcal^p)  &= \frac{1}{2} + \frac{1}{4} \norm{\sum_{i,j} p \rho_{i1j1} \ket{i}\bra{j} \otimes \hat{Z} }_1 \\&= \frac{1}{2} + \frac{p}{2} \norm{\sum_{i,j}\rho_{i1j1} \ket{i}\bra{j}}_1 ,
\end{align}
where $\hat{Z} = \kebra0-\kebra1$ is the Pauli Z matrix. One can here see, that  $\max_\rho p_{s,\rho}^{\mathcal{M}}(\Zcal, \Wcal^p)$ is achieved for any coherent superposition of input states $\kebra{01}$ and $\kebra{11}$, with $\max_\rho \norm{\sum_{i,j}\rho_{i1j1} \ket{i}\bra{j}}_1 =1$, and thus:
\begin{equation}
    \max_\rho  p_{s,\rho}^{\mathcal{M}}(\Zcal, \Wcal^p) = p_{s,\max}^{\mathcal{M}}(\Zcal, \Wcal^p) = \frac{1}{2} + \frac{p}{2}
\end{equation}

We now derive a lower bound on the optimal probability of discriminating $\Zcal$ and $\Wcal^p$ provided access to the post-measurement state by discriminating between the quantum channel representations $\map\Ical_\Zcal$ and $\map\Ical_{\Wcal^p}$ of their corresponding Lüders' instruments $\Ical_\Zcal$ and $\Ical_{\Wcal^p}$. We can write the probability of optimally discriminating the two channel with no access to entanglement and for an input state $\rho \in \Lcal(\hil{in})$, $\rho = \sum_{i,j} \rho_{ij} \ket i \bra j $ by applying the same method we did earlier:
\begin{equation}
    p_{s,\text{LOCC}\,\rho}^{\Ical}(\Zcal, \Wcal^p) = \frac{1}{2} + \frac{1}{4} \norm{\map\Ical_{\Wcal^p}(\rho)-  \map\Ical_\Zcal (\rho) }_1,
\end{equation}

We replace again all terms with their corresponding expressions and get the following expression:
\begin{equation}
\begin{split}
    p_{s,\text{LOCC}\,\rho}^{\Ical}(\Zcal, \Wcal^p) = \frac{1}{2} + \frac{1}{4} || p\rho_{11} |10\rangle\langle10|+ \sqrt{p} \rho_{01}|00\rangle\langle10| \\ + \sqrt{p} \rho_{10}|10\rangle\langle00|-p\rho_{11}|11\rangle\langle11| \,||_1,
\end{split}
\end{equation}

The operator inside the trace norm being hermitian (self-adjoint), we can compute it by summing over the absolute eigenvalues, $\norm{A}_1 = \sum_i |\lambda_i|$, which in our case are:
\begin{align}  
    \lambda_0&= 0 \\
    \lambda_1&=  \frac{\sqrt{p}}{2} \left( \rho_{11} - \sqrt{4\rho_{01}\rho_{10} +\rho_{11}^2}\right) \\
    \lambda_2&=  \frac{\sqrt{p}}{2} \left( \rho_{11} + \sqrt{4\rho_{01}\rho_{10} +\rho_{11}^2}\right) \\
    \lambda_3&= -p\rho_{11} 
\end{align}

Thus:
\begin{equation}
      p_{s,\text{LOCC}\,\rho}^{\Ical}(\Zcal, \Wcal^p) = \frac{1}{2} + \frac{1}{4}\left(\sqrt{p}\sqrt{4\rho_{01}\rho_{10} +\rho_{11}^2} +p\rho_{11} \right)
\end{equation}

By imputing the state $\kebra 1$, we can set a lower bound  $ p_{s,\text{LOCC}\,\kebra 1}^{\Ical}(\Zcal, \Wcal^p)$ on the maximum value of $ p_{s,\text{LOCC}\,\rho}^{\Ical}(\Zcal, \Wcal^p)$:
\begin{equation}
    p_{s,\text{LOCC}\,\kebra 1}^{\Ical}(\Zcal, \Wcal^p) = \frac{1}{2} + \frac{1}{4}\left( p + \sqrt{p}\right).
\end{equation}

We then have:
\begin{equation}
    \Delta(\Zcal,\Wcal^p) \geq \frac{p + \sqrt{p}}{2p} 
\end{equation}

Where for $p \rightarrow 0$, $\Delta \rightarrow \infty$, thus proving that for any $L\in \mathbb{R}$, there exists $\epsilon$, for which $\Delta(\Zcal,\Wcal^\epsilon) > L.$
\end{proof}

\section{Numerical investigation on more scenarios}

In this section we make use of the tester formalism presented in \cref{methods} to analyze further examples of discrimination tasks and illustrate the results obtainable using semidefinite programming. All code used for computing results shown in this section is openly available at~\cite{charbeleid7_github_measurement_discrimination}.

\subsection{Discrimination of a pair of trine measurements}
We start by analyzing the discrimination of the pairs of three outcome measurements defined by
\begin{equation}
    \Lambda \left( \theta, \phi \right) =  \left\{  \frac{1}{3} \left( \id + \vec{n}_j \cdot \vec{\sigma} \right) \right\}_{j=0,1,2},
\end{equation}
with $\vec{\sigma}=(\hat\sigma_X,\hat\sigma_Y,\hat\sigma_Z)$, the vector of Pauli matrices, and:
\begin{equation}
    \vec{n}_j =
    \begin{pmatrix}
        \cos\left(\theta + \frac{2\pi}{3}j \right) \\
        \sin\left(\theta + \frac{2\pi}{3}j\right) \cos\phi \\
        \sin\left(\theta + \frac{2\pi}{3}j\right) \sin\phi
    \end{pmatrix}.
\end{equation}

We identify here $\Lambda \left( 0, 0 \right)$ which is the set of three scaled projectors corresponding to the `trine states' $ \ket {\psi_j}= \frac{1}{\sqrt2} \left(\ket{0} + e^{\frac{2\pi i}{3} j }\ket1\right)$, $j=0,1,2$ lying equiangularly in the XY plane of the Bloch sphere. Noting that all other $\Lambda \left( \theta, \phi \right)$ correspond to Bloch rotations of these operators, we can take into account all possible pairwise discrimination tasks by discriminating $\Lambda \left( 0, 0 \right)$ and other elements $\Lambda \left( \theta, \phi \right)$. This does not lead to any loss of generality as any dichotomic discrimination of measurements $\Lambda \left( \theta, \phi \right)$ and $\Lambda \left( \theta', \phi' \right)$ can  be mapped to a discrimination of $\Lambda \left( 0, 0 \right)$ and $\Lambda \left( \theta'', \phi'' \right)$ via application of a global unitary. Sampling from $\Lambda(\theta,\phi)$ and using semidefinite programming, we can numerically solve the tester formulation of this problem and get the results shown in \cref{TrineMeas} and \cref{TrineInst} for the optimal success probabilities of the measurement and instrument discrimination tasks. We can additionally compute the instrument advantage for this task and obtain the graph in \cref{TrineIvM}.

 \begin{figure}[ht!]
    \centering
    \hspace{30pt}
    \includegraphics[width=0.5\textwidth]{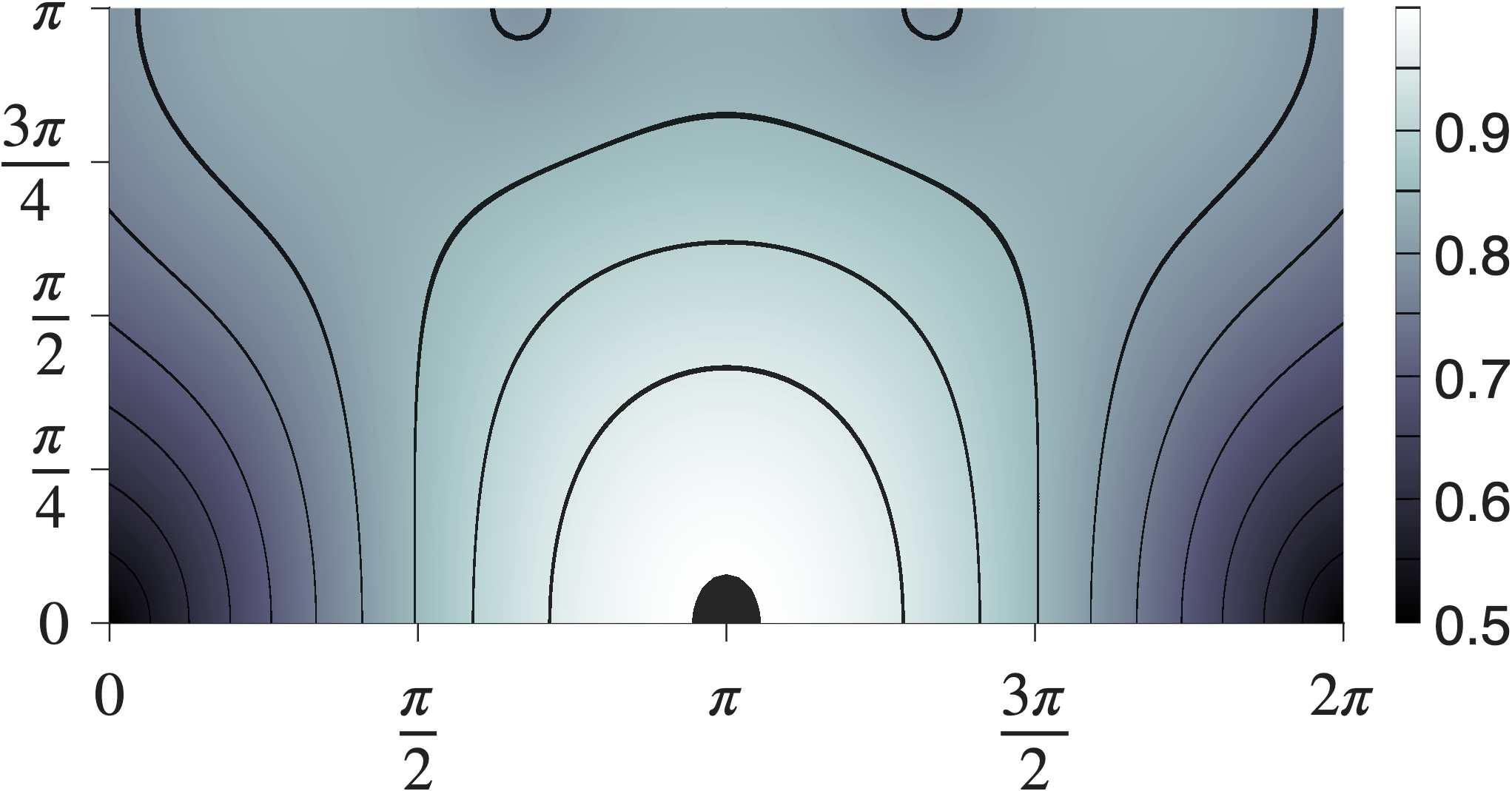}
    \caption{Success probability of discriminating measurements $\Lambda(0,0)$ and $\Lambda(\theta,\phi)$. On the vertical axis: $\phi$ values in radians. On the horizontal axis: $\theta$ values in radians. The color map represents the success probability with each of the small black separator lines covering a difference of 0.002 .}
    \label{TrineMeas}
\end{figure}

\begin{figure}[ht!]
    \centering
    \hspace{30pt}
    \includegraphics[width=0.5\textwidth]{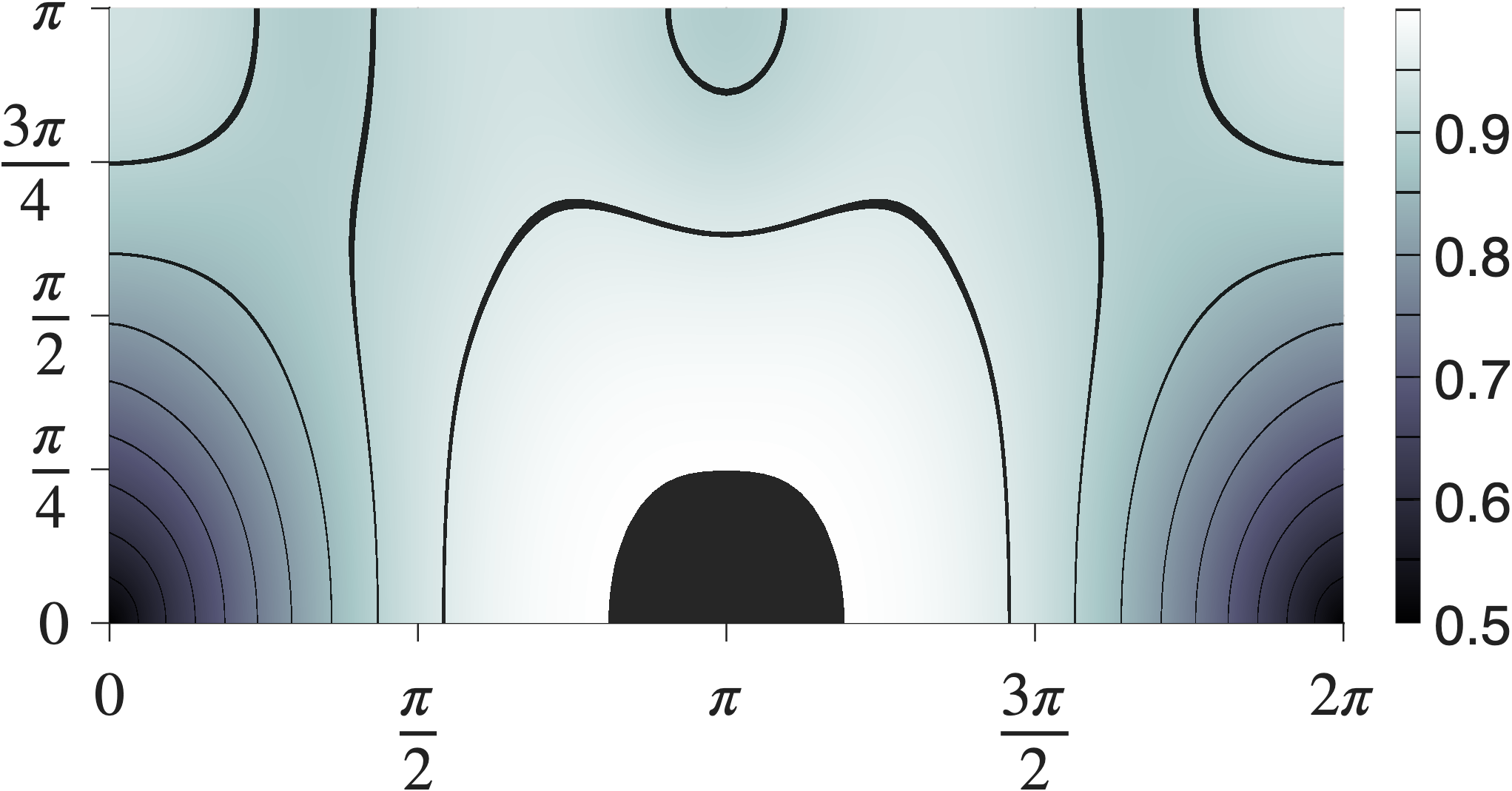}
    \caption{Success probability of discriminating Lüders' Instruments Corresponding to $\Lambda(0,0)$ and $\Lambda(\theta,\phi)$. On the vertical axis: $\phi$ values in radians. On the horizontal axis: $\theta$ values in radians. The color map represents the success probability with each of the small black separator lines covering a difference of 0.002 .}
    \label{TrineInst}
\end{figure}

\begin{figure}[ht!]
    \centering
    \hspace{30pt}
    \includegraphics[width=0.5\textwidth]{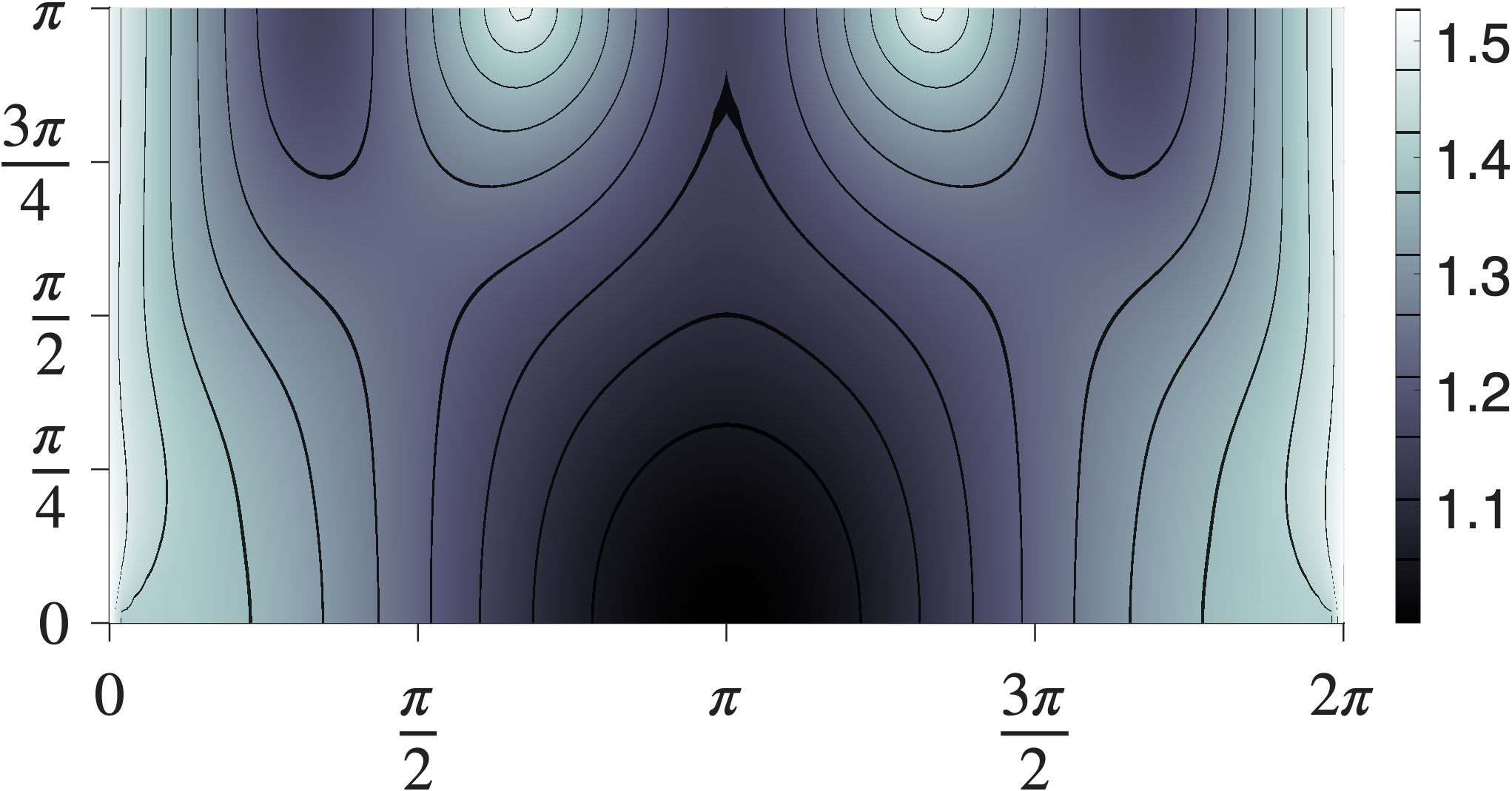}
    \caption{Instrument advantage for the task of discriminating $\Lambda(0,0)$ and $\Lambda(\theta,\phi)$. On the vertical axis: $\phi$ values in radians. On the horizontal axis: $\theta$ values in radians. The color map represents the instrument advantage with each of the small black separator lines covering a difference of 0.00165 .}
    \label{TrineIvM}
\end{figure}

\subsection{Discrimination of $\Zcal$ and $\Wcal^p$}
In the current subsection we will numerically illustrate the result obtained in \cref{theo2} as well as a few related examples. We define the measurements $\Wcal(\theta,p) := \{\kebra\theta +p\kebra{\theta^\perp} ,(1-p)\kebra{\theta^\perp} \}$, where $\ket{\theta} = \cos{\frac{\theta}{2}}\ket0 +\sin{\frac{\theta}{2}}\ket1$ and $\theta^\perp = \pi-\theta$ as rotations in the ZX plane of the measurement $\Wcal^p = \Wcal(0,p)$ considered in \cref{UnboundedAdvantageForCertainTasks}. This set of measurements can be interpreted as a model of error in the $\Zcal$ measurement, with the $p$ parameter representing noise or losses and the $\theta$ parameter representing basis misalignment in the measurement apparatus. We use semidefinite programming to discriminate $\Zcal$ and a sampling of elements from $\Wcal(\theta,p)$, which can be seen as a rudimentary model of error detection. The results for the measurement discrimination and instrument discrimination success probabilities are shown in \cref{WeakMeas} and \cref{WeakInst} respectively. The instrument advantage for the discrimination of $\Zcal$ and $\Wcal^p = \Wcal(0,p)$ is shown in \cref{WeakIvM} where the function can be seen to diverge.

\begin{figure}[ht!]
    \centering
    \hspace{30pt}
    \includegraphics[width=0.5\textwidth]{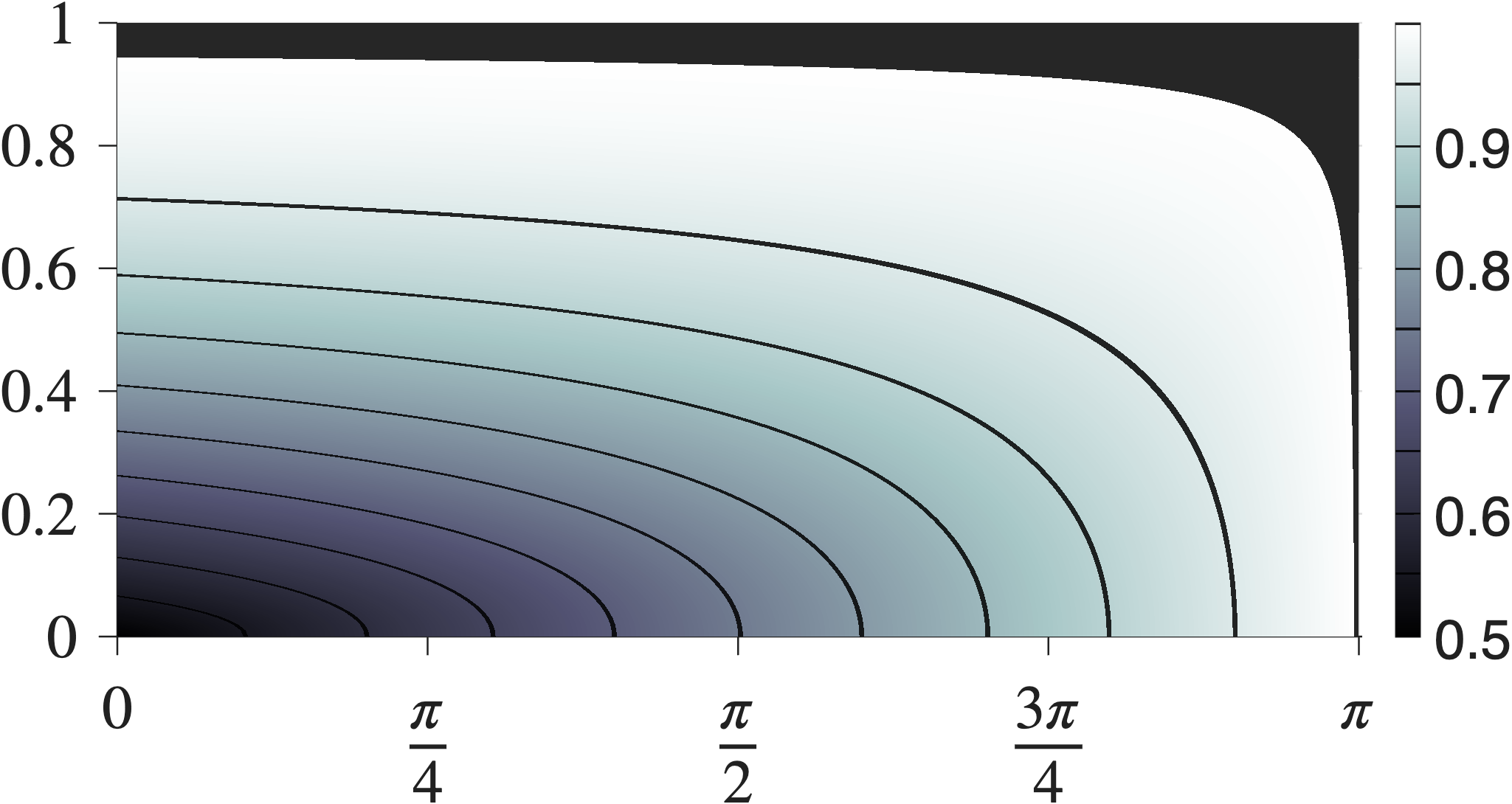}
    \caption{Success probability of discriminating measurements $\Zcal$ and $\Wcal(\theta,p)$. On the vertical axis: the parameter $p$. On the horizontal axis: $\theta$ values in radians. The color map represents the success probability with each of the small black separator lines covering a difference of 0.002 .}
    \label{WeakMeas}
\end{figure}

\begin{figure}[ht!]
    \centering
    \hspace{30pt}
    \includegraphics[width=0.5\textwidth]{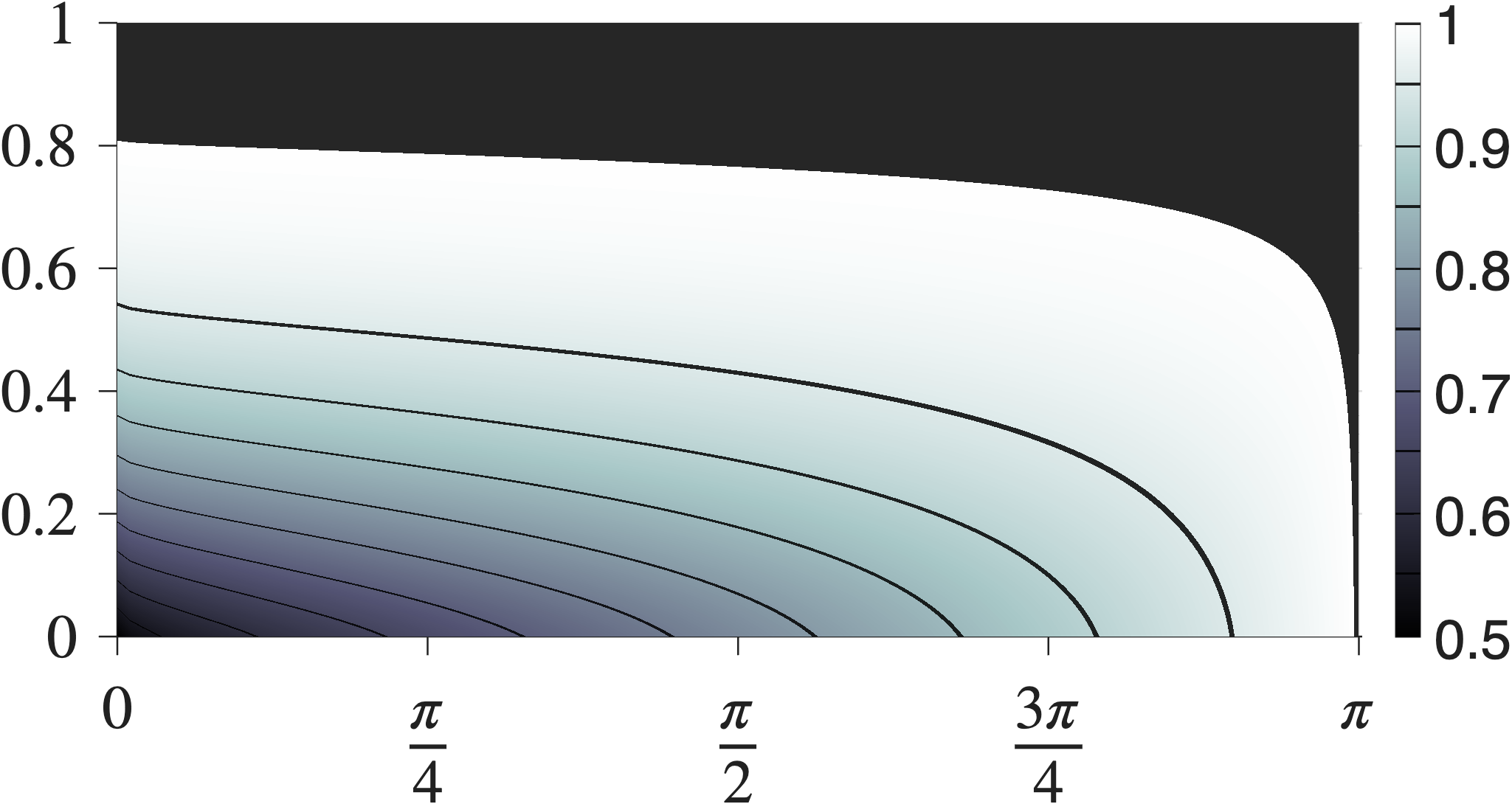}
    \caption{Success probability of discriminating Lüders' instruments corresponding to the measurements $\Zcal$ and $\Wcal(\theta,p)$. On the vertical axis: the parameter $p$. On the horizontal axis: $\theta$ values in radians. The color map represents the success probability with each of the small black separator lines covering a difference of 0.002 .}
    \label{WeakInst}
\end{figure}

\begin{figure}[ht!]
    \centering
    \hspace{30pt}
    \includegraphics[width=0.5\textwidth]{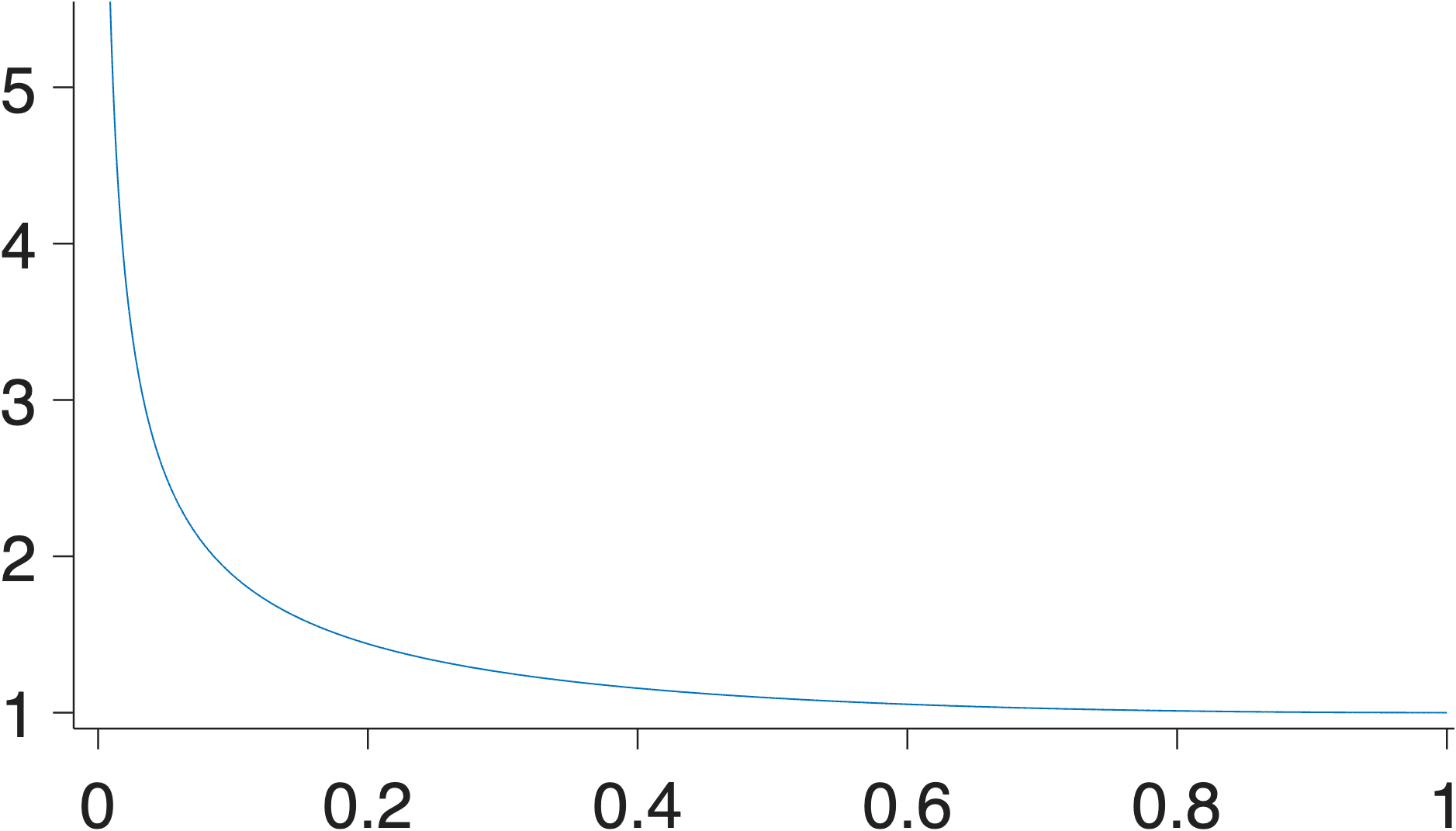}
    \caption{Instrument advantage for the discrimination of measurements $\Zcal$ and $\Wcal(0,p)$. On the vertical axis: the instrument advantage. On the horizontal axis: value of the parameter $p$. The function can be seen to diverge at 0.}
    \label{WeakIvM}
\end{figure}

\section{Discussion}

Motivated by the potential usefulness of having access to post-measurement quantum states when discriminating quantum measurements, a resource omitted by most of previous literature, we proposed the task of Lüders instrument discrimination, a measurement discrimination where the post-measurement is taken into account. We then investigate the problem of measurement discrimination with access to the post-measurement state, and adapt tools developed for quantum channels to investigate Lüders instrument discrimination from an analytical and computational perspective.

We start by showing that, in the simplest form of the problem, namely dichotomic qubit projective Lüders instrument discrimination, the problem is equivalent to two-copy pure state discrimination, extending the results from \cite{MeasD_with_post-measurement_Manna_2025}. This may be contrasted to the measurement discrimination counterpart that is equivalent to one-copy pure state discrimination~\cite{MeasD_Sedl_k_2014}, thus concluding that having access to the post-measurement state provides a significant advantage for the task. In addition to bounding the success probability, we discuss the possible strategies one can implement to optimally approach the problem, also showing that the maximum success probability can be achieved without the use of entanglement and proposing an optimal strategy that does not require it.

Then, we quantify and formalize the advantage gained in a dichotomic measurement discrimination task from having access to the post-measurement state. We do so by comparing the diamond distances between the pair of quantum channels associated to the measurement devices and the pair of channels corresponding to their Lüders instruments. We define the `Lüders Advantage Bias' as the ratio of these two quantities, and present an explicit qubit example where the advantage bias can get arbitrarily large for certain measurements, showcasing the significance of considering the post-measurement state when discriminating quantum measurements. 
Finally we employ the numerical tools developed for quantum channel discrimination to analyze the differences between measurement discrimination and Lüders discrimination tasks in scenarios that are harder to deal with analytically. Our computational code is openly available at ~\cite{charbeleid7_github_measurement_discrimination}.

\section*{Acknowledgments}
We thank Lucas Porto for useful discussions. 
MTQ is supported by the French Agence Nationale de la Recherche (ANR) under grant JCJC HOQO-KS.
This work was supported by the programme d'investissements d’excellence IDEX of the Alliance Sorbonne Université through the Quantum Information Center Sorbonne.


%
\end{document}